\documentclass[runningheads]{llncs}
\usepackage{graphicx}
\usepackage{theorem,latexsym,graphicx,amssymb}
\usepackage{amsmath,enumerate}
\usepackage{bbm}
\usepackage{subfigure}
\usepackage{float}
\usepackage{epsfig}
\usepackage{xspace}
\usepackage{paralist}
\usepackage{enumerate}
\usepackage{cases}
\usepackage{caption}
\usepackage{multicol}
\usepackage{graphicx}
\usepackage{xcolor}
\usepackage{tikz}
\usepackage{ifthen}
\usepackage{algorithm}
\usepackage[noend]{algpseudocode}
\usepackage[shortlabels]{enumitem}

\newcommand{\AutoAdjust}[3]{\mathchoice{ \left #1 #2  \right #3}{#1 #2 #3}{#1 #2 #3}{#1 #2 #3} }
\newcommand{\Xcomment}[1]{{}}

\newcommand{\InBrackets}[1]{\AutoAdjust{[}{#1}{]}}
\newcommand{\Ex}[2][]{\operatorname{\mathbf E}_{#1}\InBrackets{#2\vphantom{E_{F}}}}
\newcommand{\Exlong}[2][]{\operatornamewithlimits{\mathbf E}\limits_{#1}\InBrackets{#2\vphantom{\operatornamewithlimits{\mathbf E}\limits_{#1}}}}
\newcommand{\ExlongOpen}[2][]{\operatornamewithlimits{\mathbf E}\limits_{#1}\left[#2\vphantom{\operatornamewithlimits{\mathbf E}\limits_{#1}}\right .}
\newcommand{\ExlongClose}[1][]{\left. \vphantom{\operatornamewithlimits{\mathbf E}\limits_{t}#1}\right]}
\newcommand{\Prx}[2][]{\operatorname{\mathbf{Pr}}_{#1}\InBrackets{#2}}

\def\prob{\Prx}

\newcommand{\dd}{\mathrm{d}}  

\newcommand{\be}{\begin{equation}}
\newcommand{\ee}{\end{equation}}

\newcommand{\argmax}{\mathop{\rm argmax}}

\newcommand{\eps}{\varepsilon}

\newcommand{\vect}[1]{\ensuremath{\mathbf{#1}}}

\newcommand{\dist}{\ensuremath{F}}
\newcommand{\disti}[1][i]{\ensuremath{F_{#1}}}

\newcommand{\dists}{\vect{\dist}}

\newcommand{\bid}{b}
\newcommand{\bids}{\vect{\bid}}
\newcommand{\bidsmi}[1][i]{\bids_{\text{-}#1}}
\newcommand{\bidi}[1][i]{{\bid_{#1}}}

\newcommand{\val}{v}
\newcommand{\vals}{\vect{\val}}
\newcommand{\valsmi}[1][i]{\vals_{\text{-}#1}}
\newcommand{\vali}[1][i]{{\val_{#1}}}

\newcommand{\util}{u}

\newcommand{\utili}[1][i]{\util_{#1}}

\newcommand{\phii}[1][i]{{\varphi}_{#1}}
\newcommand{\phipi}[1][i]{{\varphi}_{#1}^{+}}

\newcommand{\price}{p}
\newcommand{\prices}{\vect{\price}}
\newcommand{\pricei}[1][i]{{\price_{#1}}}

\newcommand{\alloc}{x}
\newcommand{\allocs}{\vect{\alloc}}

\newcommand{\alloci}[1][i]{{\alloc_{#1}}}

\newcommand{\allocopti}[1][i]{{\alloc^{*}_{#1}}}
\newcommand{\allocopts}{\allocs^*}

\newcommand{\opt}{\text{OPT}}

\newcommand{\ri}[1][]{\bar{R}_{#1}}

\newcommand{\lap}{\textsf{LAP}}
\newcommand{\la}{\textsf{LA}}

\newcommand{\upc}{{c^{+}}}
\newcommand{\lowc}{{c^{o}}}

\newenvironment{proofof}[1]{{\vspace*{5pt} \noindent\bf Proof of #1:  }}{\hfill\rule{2mm}{2mm}\vspace*{5pt}}

\newcommand{\eqdef}{\stackrel{\textrm{def}}{=}}
\newcommand{\rev}{\textsf{Rev}}

\newcommand{\exante}{\textsf{ex-ante}}

\title{Lookahead Auctions with Pooling}

\author{}
\date{}

\begin{document}

\title{Lookahead Auctions with Pooling}
%
%
\author{Michal Feldman\inst{1,3} \and
Nick Gravin\inst{2}
\and
Zhihao Gavin Tang\inst{2}\and
Almog Wald\inst{1}}
\authorrunning{Feldman et al.}
%
\institute{Tel Aviv University\\ \email{michal.feldman@cs.tau.ac.il, almog.wald@gmail.com}\\ \and
ITCS, Shanghai University of Finance and Economics 
\email{\{nikolai,tang.zhihao\}@mail.shufe.edu.cn}\\ 
\and
Microsoft Research}
\maketitle              

\begin{abstract}
A Lookahead Auction (\textsf{LA}), introduced by Ronen, is an auction format for the sale of a single item among multiple buyers, which is considered simpler and more fair than the optimal auction. 
Indeed, it anonymously selects a provisional winner by a symmetric ascending-price process, and only then uses a personalized posted price.
A \textsf{LA} auction extracts at least $1/2$ of the optimal revenue, even under a correlated value distribution.
This bound is tight, even for $2$ buyers with independent values.
We introduce a natural extension of \textsf{LA}, called {\em  lookahead with pooling} (\textsf{LAP}). 
A \textsf{LAP} auction proceeds as \textsf{LA}, with one difference: it allows the seller to \emph{pool} together a range of values during the ascending-price stage, and treat them the same; thus, it preserves the simplicity and fairness of \textsf{LA}.
Our main result is that this simple pooling operation improves the revenue guarantees for independent buyers from $1/2$ to $4/7$ of the optimal revenue. 
We also give a complementary negative result, showing that for arbitrary correlated priors \textsf{LAP} cannot do better than $1/2$ approximation.

\end{abstract}

\section{Introduction}
\label{sec:introduction}
Optimal auction design is important both theoretically and practically, and has been extensively studied in both economics and computer science over the last few decades~\cite{HartlineBook}. 
The scenario of a single item auction is the most fundamental setting, and serves as the basis for the design and analysis of auctions \cite{Myerson81}. 

In a single item auction, an item is sold to one of $n$ bidders. Each bidder has a value $\vali$ drawn from an underlying distribution $\disti$. The value
$\vali$ is $i$'s private information, while $\disti$ is known to all. 
A common way to sell items is through auctions. An auction receives as input the bidder values, and determines an allocation rule (who gets the item) and a payment rule (how much each bidder pays) based on the reported values. 
The utility of a bidder is her value for the item (if she wins) minus the payment she makes.

An auction is {\em dominant-strategy incentive compatible} (DSIC) if it is in the best interest of every bidder to bid her true value $\vali$, for any value profile of the other bidders.
An auction is {\em individually rational} (IR) if the utility of all bidders is non-negative. 
The seller's {\em revenue} is the payment collected by the auction. 
An auction is said to be optimal if it maximizes the expected revenue among all DSIC and IR mechanisms.

Myerson \cite{Myerson81} gave a full characterization of the revenue-optimal auction in the case where agent values are independent.
Specifically, he showed that an auction is DSIC if and only if its allocation rule is monotone (i.e., the allocation of bidder $i$ is non-decreasing in her value), and the payment is then determined uniquely (up to normalization) by the allocation rule. Given value $\vali$ drawn from $\disti$, the {\em virtual value} of bidder $i$ is defined as $\varphi_i(\vali)=\vali-\frac{1-\disti(v_i)}{f_i(v_i)}$, where $f_i$ is the derivative of $\disti$. 
The main observation of \cite{Myerson81} is that maximizing the expected revenue is equivalent to maximizing the {\em virtual welfare}, i.e., the sum of virtual values.


The simplest case is one where bidder values are identically and independently distributed according to a {\em regular} distribution, meaning that the virtual value is monotonically non-decreasing in the value. 
In this case, the optimal auction is essentially a second-price auction with a reserve price.
However, as we move slightly beyond this scenario, the optimal auction becomes less intuitive and less natural. 
For example, if valuations are non-identically distributed, then the winner may not necessarily be the bidder who placed the highest bid. 
Moreover, for non-regular distributions, some ironing process takes place, which further complicates the auction from a bidder's perspective.

Thus, the optimal auction suffers from some undesirable properties, including asymmetry among bidders and lack of simplicity, which makes it hard to apply in practice.
Indeed, even in the simplest case of a single-item auction, practical applications tend to prefer simpler and more natural auction formats over the optimal auction. 
This observation has given rise to a body of literature known as ``simple versus optimal auctions" \cite{HartlineR09} that studies the trade offs between simplicity and optimality in auction design. 
The goal is to design simple auctions that provide good approximation guarantees to the optimal revenue. 

    
\paragraph{Lookahead auctions.}
The {\em lookahead} (\la) auction is a simple auction format that has been introduced by Ronen \cite{Ronen01}. 
Ronen showed that the \la\ auction format gives at least a half of the optimal revenue, even if agent values are distributed according to a correlated joint distribution (i.e., where distributions are not independent, in contrast to the setting studied by Myerson). 

In a \la\ auction, the item can only go to the highest-value bidder $i^*$, but this bidder does not always win the item. Instead, she is offered the item for a price that equals the revenue-optimal price for the distribution $\disti[i^*]$, given the bids of all other bidders and the fact that $i^*$ has the highest value, and buys it if her value exceeds the price.

The \la\ auction can equivalently be described as follows: increase the price continuously until a single agent remains. Then, sell the item to the remaining bidder at a price that depends only on the non-top bidders. 
This is a natural ``ascending price" process, which is common in practical auctions.
One can easily observe that this auction is DSIC. Indeed, non-top bidders are never allocated, thus have no incentive to lie; and the top bidder is offered a price that does not depend on her bid, thus has no incentive to lie either.
Clearly, it is also IR since the provisional winner buys the item only if her value exceeds the price.

Moreover, unlike the optimal auction, which may treat different bidders very differently, the \la\ auction may be perceived as more fair. Indeed, the process of identifying the provisional winner is symmetric; differential treatment is applied only for that bidder and only for determining the price. 
See \cite{DebPai17} for further discussion on symmetry and discrimination in auctions. 


The $1/2$ approximation provided by the \la\ auction format is tight, even for independent values. Namely, there exists an instance with two independently distributed bidders where no \la\ auction gives more than half of the optimal revenue. 
Improvements are possible by a variant of the \la\ auction called the $k$-\la\ auction \cite{Ronen01,ChenHLW11,DobzinskiFK11}, which finds the optimal revenue that can be obtained by the top $k$ bidders (see more on this in Section~\ref{sec:related}). 

A clear advantage of the \la\ auction is its simplicity. 
As it turns out, however, \la\ entails a revenue loss that may account for up to $50\%$ of the optimal revenue.
The question that leads us in this paper is whether there is a different way to trade off simplicity for optimality, in a way that would give better revenue guarantees. 

To this end, we introduce a variant of the \la\ auction, which we call {\em lookahead with pooling} (\lap).
\lap\ is essentially a \la\ auction with the option of ``pooling" some types together. 
In particular, we allow the mechanism to pool together an interval of values and treat them as the same type. I.e., the mechanism is allowed to allocate the item to any of the bidders within the same interval, similar to ironing in the Myerson's auction when it is applied to a single bidder.
\lap\ preserves many of the merits of \la, including simplicity and fairness, while providing better approximation guarantees. 
In particular, while \la\ cannot give a better approximation than $1/2$, even for a simple scenario with two independent bidders, \lap\ gives a better approximation for any number of independent bidders. 

The following example demonstrates how pooling works, and gives some intuition as to how it increases the revenue compared to \la.

\begin{example}
	Consider a setting with 2 bidders, where $\vali[1]=1$, and 
	$$
	\vali[2]=
	\begin{cases}
	1+\eps, & \mbox{ with probability } 1-\eps\\
	1/\eps, & \mbox{ otherwise}\\
	\end{cases}
	$$
	The optimal revenue is 2, which is obtained by offering a price of $1/\eps$ to bidder 2, and offering a price of 1 to bidder 1 in case bidder 2 rejects.
	The maximum revenue that can be obtained by \la\ is $1+\eps$, since \la\ always sells to bidder 2 and $\val \cdot \prob{\vali[2] \geq \val} \leq 1 + \eps$ for every $\val$.
	However, \lap\ can do better by pooling together the interval $[1,1/\eps]$.
	I.e., we do not differentiate different values within the interval of $[1,1/\eps]$. One can also interpret it as a discrete jump of the ``ascending price'' from $1$ to $1/\eps$.
	By doing so, the first bidder must drop after the price jump. If the second bidder also drops, our mechanism sells the item at a price of $1$ to a random uniformly chosen  bidder; if the second bidder survives, our mechanism sells the item at a price of $\frac{1+\eps}{2\eps}$. The prices are set to guarantee the truthfulness of our mechanism. See Section~\ref{sec:prelim} for a more detailed discussion. In summary, the seller's revenue is 1 with probability $1-\eps$ and  $\frac{1+\eps}{2\eps}$ with probability $\eps$, giving a revenue of $1.5-\eps/2$.
	One can verify that this is the optimal \lap.
	This example also shows that \lap\ cannot give a better approximation than $3/4$, even for settings with two independent bidders.
\end{example}

\subsection{Our Results}

Our main result is an improved bound on the approximation ratio of \lap\ auctions for the case of independent bidders.

\vspace{0.1in}
\noindent {\bf Theorem:}
For any setting with an arbitrary number of independent bidders, \lap\ achieves a $\frac{4}{7}$-approximation to the optimal revenue. 
\vspace{0.1in}

Interestingly, the \lap\ auction that provides this guarantee is a special case of $2$-\la\ (i.e., $k$-\la\ for $k=2$), which exhibits more symmetry than a general $2$-\la.
(In general, \lap\ may not be $k$-\la, it just so happens that the \lap\ we use adheres to a $2$-\la\ format.)
We complement our main result with a negative result, showing that for correlated values,  \lap\ does not improve over the approximation ratio achievable by \la, namely $1/2$.

\vspace{0.1in}
\noindent {\bf Theorem:}
There exists an instance with two correlated bidders, where no \lap\ mechanism achieves better than $\frac{1}{2}+o(1)$ approximation to the optimal revenue.

\subsection{Related Work}
\label{sec:related}
Dobzinski, Fu and Kleinberg~\cite{DobzinskiFK11} extended the lookahead auction to $k$-lookahead auctions for correlated bidders, and proved an approximation ratio of $\frac{2k-1}{3k-1}$. Subsequently, Chen, Hu, Lu and Wang~\cite{ChenHLW11} improved the approximation ratio to $\frac{e^{1-1/k}}{e^{1-1/k}+1}$. Dobzinski and Uziely~\cite{DobzinskiU18} showed that any DSIC single-item auction that is forbidden to allocate item to the lowest bidder has a constant factor gap $1-\frac{1}{e}$ compared to the optimal auction, even when the number of bidders $n$ goes to infinity. Their result implies that $k$-lookahead auctions with $k\to\infty$ does not approach the revenue of the optimal auction. Bei, Gravin, Lu and Tang~\cite{BeiGLT19}
considered single-item auction in the correlation robust framework. Among other auction formats, they study lookahead auctions and obtain computational and approximation results in the correlation-robust framework.

While single-item auctions with independent values are well understood, scenarios with correlated values are much less understood. 
Cremer and McLean \cite{CremeMcLean} introduce a mechanism that extracts full surplus in revenue for settings with a correlated joint distribution. However, this auction is considered highly non-practical, and satisfies only {\em interim} IR, as opposed to ex-post IR, meaning that IR holds only in expectation over others' values. 

Papadimitriou and Pierrakos \cite{PapadimitriouP11} study the computational complexity of optimal deterministic single-item auctions with a general joint prior distribution. They provide an inapproximability result for $n\ge 3$ bidders (and an efficient auction for 2 bidders). This is another indication of the complexity of correlated distributions. 

\section{Preliminaries}
\label{sec:prelim}
A set $N$ of $n$ potential buyers participate in a \emph{single item} auction. Each buyer $i\in N$ has a private value $\vali$ for obtaining the item. By the revelation principle, we restrict our attention to truthful sealed bid auctions: that is, each bidder $i$ submits a sealed bid $\bidi$ to the auctioneer who upon receiving all bids $\bids=(\bidi[1],\ldots,\bidi[n])$ decides on the allocation $\allocs(\bids)=(\alloci[1],\ldots,\alloci[n])$ and payments $\prices(\bids)=(\pricei[1],\ldots,\pricei[n])$. Each allocation function $\alloci(\bids)\in[0,1]$ indicated the probability that bidder $i$ gets the item and since at most one bidder can win the item, we have $\sum_{i\in N}\alloci(\bids)\le 1$. 

Throughout most of the paper, we assume that all values $\vali$ are \emph{independently} drawn from \emph{known prior} distributions $\vali\sim\disti$ and $\vals=(\vali[1],\ldots,\vali[n])\sim\dists=\disti[1]\times\ldots\times\disti[n]$. We assume that each $\disti$ is a continuous distribution supported on a bounded interval $[0,B]$.
We use $\disti$ to refer to the cumulative distribution function (CDF) of the distribution $\disti(t)=\prob[\vali]{\vali\le t}$ and use $f_i(t)$ to denote its probability density function at $t$.
We will also consider the case of interdependent prior $\dists$ in which case we write $\vals\sim\dists$.

We are interested in dominant-strategy incentive compatible (DSIC) and individually rational (IR) mechanisms. As standard, we assume that bidders maximize their \emph{quasi-linear} utility $\utili(\vali,\bids)=\vali\cdot \alloci(\bids)-\pricei(\bids)$. 
A mechanism is DSIC if for every bidder $i$, valuation profile $\vals$ and a deviation $\vali'$ of bidder $i$  
\begin{multline*}
\utili(\vali,(\valsmi,\vali))=\vali\cdot\alloci(\vals)-\pricei(\vals)\ge \utili(\vali,(\valsmi,\vali'))= \\
\vali\cdot\alloci(\valsmi,\vali')-\pricei(\valsmi,\vali'),\quad\quad
\text{(DSIC)}
\end{multline*}
where we use standard notation 
$\valsmi=(\vali[1],\ldots,\vali[i-1],\vali[i+1],\ldots,\vali[n])$ and we similarly denote $\bidsmi=(\bidi[1],\ldots,\bidi[i-1],\bidi[i+1],\ldots,\bidi[n])$.
A mechanism is IR if for every bidder $i$ and valuation profile $\vals$,
\begin{align*}
\utili(\vali,\vals)=\vali\cdot\alloci(\vals)-\pricei(\vals)\ge 0 && \text{(IR)}
\end{align*}

Our goal is to maximize the expected revenue $\rev\eqdef\Ex[\vals]{\sum_{i\in N}\pricei(\vals)}$ of the seller (auctioneer) over truthful mechanisms (DSIC $+$ IR).  

\paragraph{Myerson's Lemma.} Myerson~\cite{Myerson81} characterized DSIC mechanisms and described revenue maximizing auction for independent prior distribution\footnote{Myerson's results hold even for a weaker requirement of Bayesian Incentive Compatibility (BIC) for the agents when the prior distributions are independent.} $\dists$: 
\begin{itemize}
\item A mechanism $(\allocs,\prices)$ is DSIC if and only if
\begin{enumerate}[(a)]
    \item $\alloci(\vali,\valsmi)$ is monotonically non-decreasing in $\vali$ for each bidder $i$ and valuations $\valsmi$.
    \item The payment is fully determined by the allocation rule, and is given by 
$$
\pricei(\vals) = \vali\cdot\alloci(\vals) - \int_{0}^{\vali} z\cdot \alloci(z,\valsmi)\;\dd z.
$$
\end{enumerate}
\item The optimal mechanism selects the bidder with maximal non-negative virtual value $\phii(\vali)$ and the expected revenue is
$$
\opt=\Ex[\vals]{\max_{i\in N}(\phipi(\vali))},
$$
where $\phipi(\vali)\eqdef\max\{\phii(\vali),0\}$ and $\phii(\vali)\eqdef\vali-\frac{1-\disti(\vali)}{f_i(\vali)}$ is the virtual value for \emph{regular} distribution $\disti$ (for irregular distribution $\disti$ an ironed virtual value is used instead).
\end{itemize}

A distribution $F$ is called {\em regular} if the virtual value function $\phii(\vali)$ is non-decreasing in $\vali$.
For the case of irregular distribution $F$, Myerson defined an \emph{ironing} procedure on the  \emph{revenue curves} that allows to transform irregular distribution into a regular one.

\paragraph{Revenue curves \& Ironing.} 
Consider a bidder with value drawn from distribution $F$.
Given a price $p$, the buyer buys the item with probability $1-F(p)$, leading to an expected revenue of $p\cdot(1-F(p))$. 
Let $q(p)\eqdef 1-F(p)$ denote the selling probability at price $p$, and call it the {\em quantile} of $p$. It is convenient to plot the revenue in quantile space, i.e., as a function of the quantile $q \in [0,1]$. We get $R(q)= q \cdot F^{-1}(1-q)$ which is called the {\em revenue curve} of $F$. In other words, $R(q)$ is the revenue of a selling strategy that sells with ex-ante probability $q$. Without loss of generality, one may assume that $R(0)=R(1)=0$. The corresponding value of a quantile $q$ is $v(q)=F^{-1}(1-q)$ and the derivative of $R(q)$ at $q$ is the virtual value $\varphi(v)=v-\frac{1-F(v)}{f(v)}$ of $v$.
Thus, for regular distributions, the revenue curve is concave. For irregular distributions we additionally apply ironing procedure to the revenue curve.

One way to achieve an ex-ante selling probability of $q$ is to set a price $p=F^{-1}(1-q)$. 
However, achieving an ex-ante selling probability $q$ is possible also by randomizing over  different prices. 
For example, we may randomize between two quantiles $q_1$ and $q_2$, choosing the price corresponding to $q_1$ with probability $\alpha$, and the price corresponding to $q_2$ with probability $1-\alpha$. Choosing $q_1,q_2,\alpha$ such that $\alpha\cdot q_1+(1-\alpha)\cdot q_2=q$ also gives quantile $q$.
The revenue obtained by this randomized pricing is $\alpha \cdot R(q_1)+(1-\alpha) \cdot R(q_2)$.

If $F$ is a regular distribution, then the revenue curve $R(q)$ is concave, and the maximum possible revenue for an ex-ante selling probability $q$ is achieved by a deterministic price of  $p=F^{-1}(1-q)$. 

If, however, $F$ is irregular, then $R(q)$ is not a concave function. 
Let $\bar{R}(q)$ denote the maximum revenue achievable for a quantile $q$ (possibly by a random pricing). 
$\bar{R}(q)$ corresponds to the concave envelope of the revenue curve $R(q)=q\cdot F^{-1}(1-q)$. I.e., $\bar{R}(q)=\max_{\alpha,q_1,q_2} \{\alpha \cdot R(q_1)+(1-\alpha)R(q_2) \mid \alpha\cdot q_1+(1-\alpha)\cdot q_2=q\}$. The \emph{ironed virtual value} is defined as the derivative of the concave function $\bar{R}(q)$ at $q$ (similar to the virtual value in the regular case, which is equal to the derivative of $R(q)$ at $q$).
Note that the concave envelope $\bar{R}(q)$ of a single-parameter function $R(q)$ uses at most $2$ points of $R(q)$: we find a few disjoint intervals
$[a_1,b_1], \ldots, [a_k,b_k]$ and let $\ri(q)$ be the linear combination of the endpoints $(a_j, R(a_j)), (b_j,R(b_j))$ for the region $q\in [a_j,b_j]$; and $\ri(q)=R(q)$ for all other $q$.
Thus the ironed virtual value along each such interval $[a_j,b_j]$ is a constant (the derivative of a linear function $\ri(q)$). 
This means that for an irregular distribution $\disti=F$ for bidder $i$, all values $\vali\in[p(a_j),p(b_j)]$ within an ironed interval $q\in[a_j,b_j]$ are treated in the same way by the revenue optimal auction. 

\paragraph{Ex-ante optimal revenue.}
A common way to derive an upper bound on the optimal revenue is by the {\em ex-ante} relaxation. 
Given concave revenue curves $R_1 , \ldots , R_n$ (of regular distributions, or ironed ones $\bar{R}_1, \ldots, \bar{R}_n$), the ex-ante optimal revenue is given by the following program (referred to as the ex-ante program):
$$ \max_{\allocs} \sum_{i\in N} R_i(\alloci)\quad\quad\quad \text{s.t. } \sum_{i\in N} \alloci \le 1,\quad\forall i\in N ~~ \alloci\ge 0 
$$
Let $\allocopts=(\allocopti[1],\ldots,\allocopti[n])$ denote the optimal solution to this program.
The optimal ex-ante revenue is an upper bound on the optimal revenue. We refer the reader to \cite{HartlineR09} for more details.     
    
\paragraph{Lookahead Auctions.} A Lookahead Auction (LA) introduced by Ronen~\cite{Ronen01} is a DSIC and IR mechanism that can only allocate to the highest bidder. Every such auction consists of two stages: (i) the seller starts with a $c=0$ cutoff price and keeps continuously increasing it until all but one bidder $i^*$ drop from the competition, (ii) once only $i^*$ remains active, the seller sets a final take-it-or-leave price $\price(i^*,\valsmi[i^*])\ge c$ for $i^*$. The best LA achieves $1/2$-approximation (i.e., at least a half of the optimal revenue), even for the case of correlated prior distribution $\dists$, and the approximation ratio $1/2$ is tight even for independent priors $\dists=\prod_{i\in N}\disti$. The following LA maximizes revenue for the case of independent priors. 

\begin{definition}[Lookahead auction (\la)] Continuously and uniformly raise the cutoff price $c$ for all bidders $i\in N$
until only one bidder $i$ remains with $\vali\ge c$. Then, offer a final take-it-or-leave price $p\ge c$ to bidder $i$ based on $\disti$ and $c$, where the price $p$ is chosen to maximize the expected revenue $p = \argmax_{r\ge c} (r \cdot \prob{\vali \ge r \mid \vali\ge c})$.
\end{definition}

We next introduce a new variant of \la, called {\em lookahead with pooling} (\lap).
Similar to the treatment of irregular distributions in Myerson's auction, we consider the effect of ironing on the \la\ auction format. Namely, we allow the auctioneer, in addition to the standard continuous increments, to perform discrete jumps of the cutoff price $c$ from some $\lowc$ to some  $\upc>\lowc$ (i.e., \emph{pool together} all types $\val\in[\lowc,\upc)$) at any moment of the auction.
Consider such a jump, and let $S$ denote the set of currently active bidders who (simultaneously) drop as a result of this jump. 
If $S$ contains all currently active bidders, then the item is allocated to a random bidder $i\in S$ for a price of  $\pricei=\lowc$. 
If all but one active bidder $i^*$ drop, then bidder $i^*$ is offered a \emph{menu} of two options: (i) get the item $\alloci[i^*]=1$ at a high price $\pricei=p\cdot\frac{|S|-1}{|S|}+\lowc\cdot\frac{1}{|S|}$, where $
p = \argmax_{r\ge \upc} (r \cdot \prob{\vali \ge r \mid \vali\ge \upc})
$, or (ii) double back on the high bid (claim $\lowc\le \vali[i^*]<\upc$) and get $\alloci[i^*]=\frac{1}{|S|}$ ($S$ is the set of active bidders at $c=\lowc$) at a low price of $\pricei=\lowc$; note that the other bidders $i\in S$ get nothing $\alloci=0$ in this case. 
This family is called lookahead with pooling (\lap). 

We give below the formal definition of LAP for the general case of possibly correlated prior distribution.
Similarly to \la, a \lap\ auction maintains a cutoff price $c$ at every point in time, and raises it as the auction proceeds.

\begin{definition}[Lookahead auction with Pooling (\lap)] At each point in time, raise the cutoff price $c$ uniformly for all bidders in one of two possible ways: (a) continuously, as in LA; or (b) in a discrete jump from the current cutoff $\lowc$ to some $\upc>\lowc$, as long as the set of remaining buyers $S(c)\eqdef\{i \mid \vali\ge c \}$ consists of more than one bidder, i.e., $|S(c)|\ge 2$. 
Finally, after the last increment, either no bidder, or a single bidder accepts the cutoff price. If the last remaining bidder $i$ was determined in a continuous phase (a) of raising $c$, then $\{i\}=S(c)$ and we proceed as in LA (offer $i$ the revenue maximizing take-it-or-leave price $p = \argmax_{r\ge c} (r \cdot \prob{\vali \ge r \mid \vali\ge c})$). If we stop after a discrete jump (b) from $\lowc$ to $\upc>\lowc$ (i.e., $|S(\upc)|\le 1$), then 
\begin{itemize}
\item If $S(\upc)=\emptyset$, then the item is sold to a bidder $i$ chosen uniformly at random from  $S(\lowc)$ at a price $\lowc$ (i.e., for every $i\in S(\lowc)~~\alloci=\frac{1}{|S(\lowc)|}, \pricei=\alloci\cdot \lowc=\frac{\lowc}{|S(\lowc)|}$).
\item If $|S(\upc)|=1$, then the single bidder $i$ in $S(\upc)$ may choose one of the following two options:
\begin{enumerate}
\item Get the item at a price $\pricei=\lowc\cdot\alloci$ with probability $\alloci=\frac{1}{|S(\lowc)|}$.
(As if $S(\upc)=\emptyset$.)
\item Get the item with probability $\alloci=1$ at a price $\pricei= \frac{1}{|S(\lowc)|}\cdot \lowc + \frac{|S(\lowc)|-1}{|S(\lowc)|}\cdot r,$ where $r=\argmax_{t\ge \upc} \left(t \cdot \prob[\vals]{\vali \ge t ~\mid~ \vali\ge \upc, \valsmi}\right)$.
\footnote{Equivalent description: post a random take-it-or-leave price $\lowc$ to $i$ with probability $\frac{1}{|S(\lowc)|}$, with remaining probability post the revenue maximizing price $r$ to $i$ for the posterior distribution of $\vali$ given $\valsmi$ and that $\vali\ge \upc$.} 
\end{enumerate}
\end{itemize}
\end{definition}

For the case of independent prior $\dists=\prod_{i\in N}\disti$, we have 
$$r=\argmax_{t\ge \upc} \left(t \cdot \prob{\vali \ge t ~\mid~ \vali\ge \upc}\right).$$

The following theorem shows that \lap\ is truthful. The proof is deferred to Appendix~\ref{sec:truthful-lap}.
\begin{theorem}
\label{th:LAP_truthful}
Any $\lap$ auction is DSIC and IR.
\end{theorem}

\section{\lap\ for Independent Valuations}
\label{sec:lap}
In this section we prove our main result (Theorem~\ref{thm:lap_regular}).
We first prove the theorem under the assumption that all value distributions $\disti$ are regular. In Subsection~\ref{sec:irregular}, we show how to extend this result to irregular distributions.

\begin{theorem}
\label{thm:lap_regular}
For any instance with (any number of) independent bidders, $\lap$ achieves a $\frac{4}{7}$-approximation to the optimal revenue.
\end{theorem}
\begin{proof}
To prove our main theorem, we first consider a hypothetical scenario with two real bidders and one dummy bidder. The dummy bidder has a known value $v$, and the two real bidders have independent values $\vali[1],\vali[2] \geq v$. This scenario naturally appears if we keep continuously increasing the cutoff price $c$ in a $\lap$ until exactly two bidders remain in $S(v)=\{i:\vali\ge v\}$. In this case, the $\lap$ can only sell to the two ``real'' bidders and it also has to adhere to the $\lap$ format. On the other hand, it is natural to upper bound the benchmark (the optimal mechanism for all $n$ bidders) with the optimal mechanism that can allocate to the two real bidders and also the dummy bidder with deterministic value $v$. 
The following lemma (Lemma~\ref{lem:3-bidders}) shows the existence of a \lap\ auction for the two-bidder scenario that gives a $\frac{4}{7}$-approximation with respect to the three-bidder scenario.

\begin{lemma}
 \label{lem:3-bidders}
There is a $\lap$ mechanism $\textsf{M}$ 
for the instance with two bidders $A,B$ with independent values of at least $v$ that obtains a $\frac{4}{7}$ approximation against the best revenue that can be obtained in an instance with bidders $A,B$ and an additional bidder $C$ with deterministic value $v$.
\end{lemma}

We first show how Lemma~\ref{lem:3-bidders} implies Theorem~\ref{thm:lap_regular}, and then prove the lemma. 
Our $\lap$ mechanism proceeds as follows: 

\begin{enumerate}
    \item Increase continuously the cutoff price $c$ until exactly two bidders $i,j\in N$ remain.
    \item Apply the $\lap$ mechanism $\textsf{M}$ that is guaranteed by Lemma~\ref{lem:3-bidders} to $i,j$, with $v$ equal to the last cutoff price $c$.
\end{enumerate}

We use $\textsf{M}(F_A,F_B,v)$ to denote the revenue generated by $\textsf{M}$ when the value distributions of $A$ and $B$ bidders are $F_A,F_B$, conditioned on the fact that their values are at least $v$.
We have the following.

\begin{multline*}
\opt = \Exlong[\vals]{\max_\ell \varphi_\ell^+ (v_\ell)} 
 = \sum_{i\ne j\ne k} \ExlongOpen[v_k]{\Prx{v_i,v_j \ge v_k} \cdot \prod_{m \ne i,j,k} \Prx{v_m < v_k} \cdot} \\ 
 \Exlong[\vals_{\text{-}k}]{\max_{\ell} \varphi_\ell^+(v_\ell) \Big| v_i,v_j \ge v_k, v_m < v_k, \forall m \ne i,j,k} \ExlongClose[\prod_{m \ne i,j,k}] 
 \le \sum_{i,j,k} \ExlongOpen[v_k]{\Prx{v_i,v_j \ge v_k} \cdot \vphantom{\prod_{m \ne i,j,k}}}\\ 
\prod_{m \ne i,j,k} \Prx{v_m < v_k} \cdot \Exlong[v_i,v_j]{\max_{\ell} \left( \varphi_i^+(v_i), \varphi_j^+(v_j), v_k \right) \Big| v_i,v_j \ge v_k}\ExlongClose[\prod_{\ell \ne i,j,k}] 
\end{multline*}
\begin{multline*}
 \le \sum_{i,j,k} \Exlong[v_k]{\Prx{v_i,v_j \ge v_k} \cdot \prod_{m \ne i,j,k} \Prx{v_m < v_k} \cdot \frac{7}{4} \cdot \textsf{M} \left( F_i |_{v_i \ge v_k}, F_j |_{v_j \ge v_k}, v_k \right)}\\ 
 = \frac{7}{4} \cdot \lap
\end{multline*}
The first equality follows from Myerson's theorem. For the second equality, we divide the value space according to the third highest value $v_k$ and the identity $i,j$ of the remaining bidders. In other words, $v_k$ is the value of the price $c$ when our $\lap$ mechanism switches to the second step. For the first inequality, we use the fact that the virtual value is no larger than the value, i.e., $\varphi_m^+(v_m) \le v_m \le v_k$. For the second inequality, we apply Lemma~\ref{lem:3-bidders} with respect to bidders $i,j$ and value $v_k$. Notice that by Myerson's theorem, $\Ex{\max ( \varphi_i^+(v_i), \varphi_j^+(v_j), v_k )}$ is exactly the optimal revenue when there is an additional bidder of deterministic value $v_k$. The final equality follows from the definition of our mechanism.
\end{proof}

\begin{remark}
The reduction from our main theorem to Lemma~\ref{lem:3-bidders} directly applies to irregular distributions.
Notice that for regular distributions, the virtual values with respect to $F_i$ and $F_i|_{v_i \ge v}$ are the same for all values $v_i \ge v$. We implicitly use this fact in the second inequality above. 
For irregular distributions, the inequality continues to hold in the right direction, i.e. the virtual value function only increases when we restrict the distribution to $v_i \ge v$.
\end{remark}

We next prove our main lemma.

\begin{proofof}{Lemma~\ref{lem:3-bidders}}
Consider the {\em ex-ante relaxation} for the 3-bidder scenario, with variables $x_1,x_2,x_3$ denoting the respective allocation probabilities  and $R_A(x_1), R_B(x_2)$, and $R_C(x_3)$ denoting the respective revenues obtained from the three bidders $A, B,$ and $C$.

\begin{equation}
 \underset{x_1,x_2,x_3}{\mbox{maximize}}: R_A(x_1)+R_B(x_2)+R_C(x_3) \quad \mbox{ s.t. } x_1 + x_2 + x_3 \leq 1
\end{equation}
We shall design a $\frac{4}{7}$-approximate $\lap$ mechanism against this stronger benchmark.
For regular distributions, $R_A(x_1)=x_i\cdot F^{-1}_{A}(1-x_1)$, $R_B(x_2)=x_2\cdot F^{-1}_{B}(1-x_2)$, and $R_C(x_3)=v\cdot x_3$. 
Without loss of generality, we assume that $v=1$. Then, $R_C(x_3)=x_3 \leq 1-x_1-x_2$.

We reformulate the program as follows.
\begin{equation}
	\label{ex-ante-prog}
 \max\limits_{x_1,x_2}~~ x_1 \cdot F_{A}^{-1}(1-x_1) + x_2 \cdot F_{B}^{-1}(1-x_2) + (1-x_1-x_2)  \quad \mbox{ s.t. } x_1 + x_2 \leq 1
\end{equation}
 
Let $x_1^{*}, x_2^{*}$ be the optimal solution to the ex-ante program~\eqref{ex-ante-prog}, and $p_1^{*}=
F_{A}^{-1}(1-x_1^{*}), p_2^{*}=F_{B}^{-1}(1-x_2^{*})$ be the corresponding prices. 
Let $r_1^{*}=R_A(x_1^*)=x_1^{*}\cdot p_1^{*}$ and $r_2^{*}=R_B(x_2^*)=x_2^{*} \cdot p_2^{*}$. It holds that
\begin{equation}
\label{eqn:opt_exante}
\opt_{\exante} = r_1^{*} + r_2^{*} + 1 - x_1^{*} - x_2^{*}.
\end{equation}

Without loss of generality, we assume that $p_1^* \le p_2^*$. We introduce below three \lap\ mechanisms that apply different pooling strategies. Each of these mechanisms uses pooling at most once during the course of the mechanism (and increases the cutoff price continuously at all other times). 
\begin{enumerate}
\item Continuously increase the cutoff price to $p_1^{*}$: 
    \begin{itemize}
	    \item if only bidder $A$ survives, we post a price $p_1^{*}$ to $A$;
	    \item if only bidder $B$ survives, we post a price $p_2^{*}$ to $B$;
	    \item if both $A$ and $B$ survive, pool $[p_1^{*},p_2^{*}]$. 
    \end{itemize}
\item Continuously increase the cutoff price without any pooling.
\item Either pool $[1,p_1^{*}]$ or pool $[1,p_2^{*}]$, whichever gives a higher expected revenue. 
\end{enumerate}

Let $\rev_1$, $\rev_2$, $\rev_3$ be the respective revenues of the three \lap\ mechanisms above.
\begin{claim}
\label{cl:3-bidder}
$\rev_1 \ge r_1^{*} + r_2^{*} - \frac{r_2^{*} \cdot x_1^{*} + r_1^{*} \cdot x_2^{*}}{2}$, $\rev_2 \ge \max\{r_1^{*},r_2^{*}\}$, \\
and $\rev_3 \ge \max\left\{ 1+\frac{r_1^{*}-x_1^{*}}{2}, 1+\frac{r_2^{*}-x_2^{*}}{2} \right\}$.
\end{claim}
\begin{proof}
We calculate the revenue of the first mechanism.
\begin{multline*}
\rev_1 \ge \Pr[v_1 \ge p_1^*, v_2 < p_1^*] \cdot p_1^* + \Pr[v_1 < p_1^*, v_2 \ge p_2^*] \cdot p_2^* \\
+ \left( \Pr[v_1 \ge p_1^*, v_2 \ge p_1^*] \cdot p_1^* + \Pr[v_1 \ge p_1^*, v_2 \ge p_2^*] \cdot \frac{ p_2^* - p_1^*}{2} \right) \\
= \Pr[v_1 \ge p_1^*] \cdot p_1^* + \Pr[v_1 < p_1^*, v_2 \ge p_2^*] \cdot p_2^* + \Pr[v_1 \ge p_1^*, v_2 \ge p_2^*] \cdot \frac{p_2^*- p_1^*}{2} \\
= x_1^* \cdot p_1^* + (1-x_1^*) \cdot x_2^* \cdot p_2^* + x_1^* \cdot x_2^* \cdot \frac{p_2^* - p_1^*}{2} = r_1^{*} + r_2^{*} - \frac{r_2^{*} x_1^{*} + r_1^{*} x_2^{*}}{2}. 
\end{multline*}
In the first inequality, the first two terms correspond to the revenue achieved when only one bidder survives after price $p_1^*$. In the third case when both bidders survive, we have a guaranteed revenue of $p_1^*$ and an additional revenue of at least $\frac{p_2^*-p_1^*}{2}$ if $v_2 \ge p_2^*$.\footnote{Observe that if both bidders survive after the pooling $[p_1^*,p_2^*]$, then the expected revenue we get is at least $p_2^*$. Indeed, in this case  we continuously increase the cutoff price and can post a sub-optimal price to the winner that is equal to the last cutoff price of $c$. This gives us a guaranteed payment of $c>p_2^*$.} 
The second mechanism achieves revenue greater than or equal to the revenues of the following two  mechanisms: post price $p_1^{*}$ to $A$, or post price $p_2^{*}$ to $B$. I.e., $\rev_2 \ge \max \{ r_1^*, r_2^*\}$.


For the third mechanism, if we pool $[1,p_1^{*}]$, then the revenue is at least
\[
1 + \Pr[v_1 \ge p_1^*] \cdot \frac{p_1^*-1}{2} = 1+ x_1^* \cdot \frac{p_1^*-1}{2} = 1 + \frac{r_1^*-x_1^*}{2}.
\]
The case for pooling $[1,p_2^*]$ is similar. Therefore, $\rev_3 \ge \max\left\{ 1+\frac{r_1^{*}-x_1^{*}}{2}, 1+\frac{r_2^{*}-x_2^{*}}{2} \right\}$.
\end{proof}

According to the above claim and equation~\eqref{eqn:opt_exante}, it suffices to verify that
\begin{multline}
\max \left\{ r_1^{*} + r_2^{*} - \frac{r_2^{*} \cdot x_1^{*} + r_1^{*} \cdot x_2^{*}}{2}, r_1^{*}, r_2^{*}, 1+\frac{r_1^{*}-x_1^{*}}{2}, 1+\frac{r_2^{*}-x_2^{*}}{2} \right\} \\
\ge \frac{4}{7} \cdot  \left( r_1^{*} + r_2^{*} + 1 - x_1^{*} - x_2^{*} \right)
\end{multline}
From now on, we shall treat this as an algebraic inequality and use the following properties:  1) $x_1^*+x_2^* \le 1$, and 2) $r_1^* \ge x_1^*, r_2^* \ge x_2^*$. Observe the symmetry of $r_1^*,x_1^*$ and $r_2^*,x_2^*$. We shall assume w.l.o.g. that $r_1^{*} \leq r_2^{*}$.
We prove the inequality by a case analysis.
We should recall that $v \ge 1$.

\begin{lemma}
	\label{lem:r1-greater-1}
	If $r_1^* \geq 1$, then $\max \left\{ r_1^{*} + r_2^{*} - \frac{r_2^{*} \cdot x_1^{*} + r_1^{*} \cdot x_2^{*}}{2}, r_2^{*} \right\} \ge \frac{2}{3} \cdot \left( r_1^{*} + r_2^{*} + 1 - x_1^{*} - x_2^{*} \right)$.
\end{lemma}
\begin{proof}
First, note that scaling down $r_1^{*}$ and $r_2^{*}$ by the same factor until $r_1^{*}=1$ would only make the inequality tighter, since $x_1^*+x_2^* \le 1$. Thus, we can assume that $r_1^{*}=1$. Moreover, increasing $x_1^{*}$ and decreasing $x_2^{*}$ by the same amount may only decrease the left hand side of the inequality (the constraint $r_1^*\ge x_1^*$ is satisfied as $r_1^*\ge 1\ge x_1^*+x_2^*$). Hence, it suffices to check the case when $x_2^*=0$, i.e., we are left to show
\[
\max \left\{ 1 + r_2^{*} - \frac{r_2^{*} \cdot x_1^{*}}{2}, r_2^{*} \right\} \ge \frac{2}{3} \cdot \left( 2 + r_2^{*} - x_1^{*} \right)
\]
	 
\noindent {\bf Case 1:} $x_1^{*} \cdot r_2^{*} \geq 2$. We show that $r_2^*\ge\frac{2}{3} \cdot \left( 2 + r_2^{*} - x_1^{*} \right)$. Indeed,
\[
\frac{2+r_2^*-x_1^*}{r_2^*} \leq \frac{2 + r_2^{*} - 2/r_2^{*}}{r_2^{*}} \leq \frac{3}{2},
\]
where the last inequality is equivalent to $(r_2^*-2)^2 \ge 0$.

\noindent {\bf Case 2:} $x_1^{*} \cdot r_2^{*} < 2$.
We show that $1+r_2^*-\frac{r_2^{*} \cdot x_1^{*}}{2}\ge\frac{2}{3} \cdot \left( 2 + r_2^{*} - x_1^{*} \right)$. We have
\begin{eqnarray*}
	\frac{2+r_2^{*}-x_1^{*}}{1+r_2^{*}-\frac{1}{2}r_2^{*} x_1^{*}} \leq \frac{3}{2} & \iff & 
	1 + \frac{3}{2} r_2^{*} x_1^{*} - r_2^{*} - 2 x_1^{*}\leq 0.
\end{eqnarray*}
The last inequality is linear in $x_1^{*}$ for any $r_2^{*}$. We can continuously change $x_1^{*}$ towards one of the extremes: $x_1^{*}=0$, or $x_1^{*}=1$, whichever maximizes $1 + \frac{3}{2} r_2^{*} x_1^{*} - r_2^{*} - 2 x_1^{*}$. Note that if the desired inequality holds for the updated $x_1^*$, then it must also hold for the original $x_1^*$. During this continuous change of $x_1^*$ we may get $x_1^{*}\cdot r_2^{*}=2$. In this case we stop and obtain the desired result due to the case 1 (note that  $1+r_2^*-\frac{1}{2}x_1^{*}\cdot r_2^{*}=r_2^*$ when $x_1^{*}\cdot r_2^{*}=2$). Otherwise, we reach one of the extremes $x_1^{*}=0$, or $x_1^{*}=1$ and still have $x_1^{*}\cdot r_2^{*}<2$.
For $x_1^{*}=0$, the last inequality becomes $1\leq r_2^{*}$, which is true. For $x_1^{*}=1$, the last inequality becomes  $1 + \frac{3}{2} r_2^{*} \leq r_2^{*} + 2 \iff r_2^{*} \leq 2$, which holds by virtue of case 2 (recall $x_1^{*}=1$).
\end{proof}

\begin{lemma}
	\label{lem:r1-smaller-1}
	If $r_1^{*} < 1$, then
\begin{multline*}
\max \left\{ r_1^{*} + r_2^{*} - \frac{r_2^{*} \cdot x_1^{*} + r_1^{*} \cdot x_2^{*}}{2}, r_2^{*}, 1+\frac{r_1^{*}-x_1^{*}}{2}, 1+\frac{r_2^{*}-x_2^{*}}{2} \right\} \\
\ge \frac{4}{7} \cdot  \left( r_1^{*} + r_2^{*} + 1 - x_1^{*} - x_2^{*} \right).
\end{multline*}
\end{lemma}
\begin{proof}
Note that it is w.l.o.g. to assume that $x_2^{*}=0$. Indeed, by decreasing $r_2^{*}$ ($r_2^{*}\ge x_2^{*}$) and $x_2^{*}$ by the same amount, only the first two terms $r_1^{*} + r_2^{*} - \frac{r_2^{*} \cdot x_1^{*} + r_1^{*} \cdot x_2^{*}}{2}$ and $r_2^{*}$ decrease (the first term decreases, because $\frac{x_1^*+r_1^*}{2}\le 1$, as  $1> r_1^{*} \ge x_1^*$) and nothing else changes.
Substituting $x_2=0$, it suffices to show that:
\[
\max \left\{ r_1^{*} + r_2^{*} - \frac{r_2^{*} \cdot x_1^{*}}{2}, r_2^{*}, 1+\frac{r_1^{*}-x_1^{*}}{2}, 1+\frac{r_2^{*}}{2} \right\} \ge \frac{4}{7} \cdot  \left( r_1^{*} + r_2^{*} + 1 - x_1^{*} \right).
\]

\noindent {\bf Case 1:} $r_2^{*} \leq 2$. We have 
\begin{multline*}
\frac{7}{4} \cdot \max \left\{ r_1^{*} + r_2^{*} - \frac{r_2^{*} \cdot x_1^{*}}{2}, 1+\frac{r_1^{*}-x_1^{*}}{2}, 1+\frac{r_2^{*}}{2} \right\} \\
\ge  \frac{3}{4} \cdot \left(r_1^{*} + r_2^{*} - \frac{r_2^{*} \cdot x_1^{*}}{2} \right) + \frac{1}{2} \cdot \left( 1+\frac{r_1^{*}-x_1^{*}}{2} \right) + \frac{1}{2} \cdot \left( 1+\frac{r_2^{*}}{2} \right) \\
\ge  \frac{3}{4} \cdot \left(r_1^{*} + r_2^{*} - x_1^* \right) + \frac{1}{2} \cdot \left( 1+\frac{r_1^{*}-x_1^{*}}{2} \right) + \frac{1}{2} \cdot \left( 1+\frac{r_2^{*}}{2} \right) = r_1^{*} +r_2^{*} + 1 - x_1^{*} 
\end{multline*}

\noindent {\bf Case 2:} $r_2^{*} > 2$. We have
\begin{multline*}
\frac{3}{2} \cdot \max \left\{ r_1^{*} + r_2^{*} - \frac{r_2^{*} \cdot x_1^{*}}{2}, r_2^{*} \right\} \ge  \left( r_1^{*} + r_2^{*} - \frac{r_2^{*} \cdot x_1^{*}}{2} \right) + \frac{1}{2} \cdot r_2^* \\
\geq  r_1^{*} + r_2^{*}  +\frac{1}{2} r_2^{*} (1 - x_1^{*}) \geq  r_1^{*} + r_2^{*}  + 1 - x_1^{*}
\end{multline*}
\end{proof}

Combining Lemmas~\ref{lem:r1-greater-1} and~\ref{lem:r1-smaller-1}, we conclude the proof of the main lemma. 
\end{proofof}

\subsection{Extension to Irregular Distributions}
\label{sec:irregular}
In this section we extend the proof of Theorem~\ref{thm:lap_regular} to irregular distributions.
It suffices to show that Lemma~\ref{lem:3-bidders} holds for irregular distributions. We follow the same analysis by solving the ex-ante relaxation.
\begin{equation}
 \underset{x_1,x_2}{\mbox{maximize}}: \ri[A](x_1)+\ri[B](x_2)+1-x_1-x_2 \quad \mbox{ s.t. } x_1 + x_2 \leq 1,
\end{equation}
where  $\ri[A](x_1)$ and $\ri[B](x_2)$ denote the ironed revenue curves of bidders $A$ and $B$, respectively.

Denote the optimal solution by $x_1^*,x_2^*$. The only difference compared to the regular case is that the values $\ri[A](x_1^*), \ri[B](x_2^*)$ might not correspond to a single posted price $p_1^*=F_A^{-1}(1-x_1^*)$, or  $p_2^*=F_B^{-1}(1-x_2^*)$. In fact, if both $\ri[A](x_1^*), \ri[B](x_2^*)$ correspond to posted prices $p_1^*$ and $p_2^*$, respectively, then the proof of Lemma~\ref{lem:3-bidders} for regular distributions applies intact.

Recall that each of $\ri[A](x_1^*)$ and $\ri[B](x_2^*)$ corresponds to a randomization of two prices if and only if $x_i^*$ lies on the ironed region of the corresponding revenue curve $R_A(x_1)$ or $R_B(x_2)$.

Consider the case when both $x_1^*, x_2^*$ lie strictly inside the ironed regions of their respective revenue curves $\ri[A](\cdot), \ri[B](\cdot)$. Then their derivatives $\ri[A]'(\cdot)$ and $\ri[B]'(\cdot)$ must be locally a constant around $x_1^*$, or respectively around $x_2^*$. Moreover, $\ri[B]'(x_1^*) = \ri[B]'(x_2^*)$, since otherwise, 
$$
\max \left(\ri[A](x_1^*+\eps) + \ri[B](x_2^*-\eps), \ri[A](x_1^*-\eps) + \ri[B](x_2^*+\eps) \right )> \ri[A](x_1^*) + \ri[B](x_2^*),
$$
holds for a sufficiently small $\eps>0$, violating the optimality of $x_1^*,x_2^*$.
Consequently, we can keep increasing $x_1^*$ and decreasing $x_2^*$ at the same rate until one of them hits the end of their respective ironed interval, and the objective function $\ri[A](x_1^*) + \ri[B](x_2^*)$ remains unchanged.
Therefore, we can assume that only one of $x_1^*,x_2^*$ lies inside the ironed region. Without loss of generality, we assume that it is $x_1^*$.

Furthermore, we may assume that $x_1^*+x_2^*=1$. Indeed, otherwise we can increase $x_1^*$ (when $\ri[A]'(x_1^*)>1$), or decrease $x_1^*$ (when $\ri[A]'(x_1^*)\le 1$) until $x_1^*$ escapes from its ironed region, or we get $x_1^*+x_2^*=1$.

Finally, it remains to consider the case where $\ri[A](x_1^*)$ corresponds to a randomization of two prices, $\ri[B](x_2^*)$ corresponds to a single price, and $x_1^*+x_2^* = 1$.

Let $x_1^* = \alpha \cdot x_{11}^* + (1-\alpha)\cdot x_{12}^*$, and $p_{11}^{*} = F_A^{-1}(1-x_{11}^*), p_{12}^{*} = F_A^{-1}(1-x_{12}^*)$ be the corresponding prices.
Let $r_{11}^* \eqdef x_{11}^* \cdot p_{11}^*, r_{12}^* \eqdef x_{12}^* \cdot p_{12}^*$ and $r_1^{*}\eqdef\ri[A](x_1^*)=\alpha \cdot x_{11}^{*}\cdot p_{11}^{*} + (1 - \alpha) \cdot x_{12}^{*}\cdot p_{12}^{*}$. 
Let $r_2^{*}\eqdef \ri[B](x_2^*)= x_2^{*}\cdot p_2^{*}$.
Then we have
$$
\opt_{\exante} = r_1^{*} + r_2^{*} = \alpha \cdot r_{11}^* + (1-\alpha) \cdot r_{12}^* + r_2^* 
$$

We use one of the following two $\lap$ mechanisms from the proof of Lemma~\ref{lem:3-bidders} with only slight modifications.

\begin{enumerate}
\item Continuously increase the cutoff price to $p_1^{*}$: 
    \begin{itemize}
	    \item if only bidder $A$ survives, we post a price $p_1^{*}$ to $A$;
	    \item if only bidder $B$ survives, we post a  price $p_2^{*}$ to $B$;
	    \item if both $A$ and $B$ survive, pool $[p_1^{*},p_2^{*}]$. 
    \end{itemize}
\item Continuously increase the cutoff price without any pooling.
\end{enumerate}

We modify the above mechanisms by setting $p_1^*$ to be $p_{11}^*$ or $p_{12}^*$, whichever gives a higher expected revenue. 
Denote the revenue of the two mechanisms by $\rev_1, \rev_2$. 
Following the same derivations as in Claim~\ref{cl:3-bidder}, we get that 
\begin{align*}
& \rev_1 \ge \max \left\{ r_{11}^{*} + r_{2}^* - \frac{r_2^*\cdot x_{11}^* + r_{11}^*\cdot x_2^*}{2}, r_{12}^{*} + r_{2}^* - \frac{r_2^*\cdot x_{12}^* + r_{12}^*\cdot x_2^*}{2} \right\}~; \\
& \rev_2 \ge \max \left\{ r_{11}^*, r_{12}^*, r_2^* \right\}~.
\end{align*}
Consequently,  we have
\begin{align*}
\rev_1 & \ge \alpha \cdot \left( r_{11}^{*} + r_{2}^* - \frac{r_2^*\cdot x_{11}^* + r_{11}^*\cdot x_2^*}{2} \right) + (1-\alpha) \cdot \left( r_{12}^{*} + r_{2}^* - \frac{r_2^*\cdot x_{12}^* + r_{12}^*\cdot x_2^*}{2} \right) \\
& = r_1^* + r_2^* - \frac{r_2^*\cdot x_{1}^* + r_{1}^*\cdot x_2^*}{2}~; \\
\rev_2 & \ge \max \left\{ \alpha \cdot r_{11}^* + (1-\alpha) \cdot r_{12}^*, r_2^* \right\} =  \max \left\{ r_{1}^*, r_2^* \right\} \ge x_1^* \cdot r_2^* + x_2^* \cdot r_1^*~.
\end{align*}
Finally, we have
\begin{multline*}
\lap \ge \frac{2}{3} \left( \rev_1 + \frac{1}{2} \cdot \rev_2\right) \ge \frac{2}{3} \left( r_1^* + r_2^* - \frac{r_2^*\cdot x_{1}^* + r_{1}^*\cdot x_2^*}{2} + \frac{x_1^* \cdot r_2^* + x_2^* \cdot r_1^*}{2} \right) \\
= \frac{2}{3} \left( r_1^* + r_2^* \right) = \frac{2}{3} \opt > \frac{4}{7} \opt~.
\end{multline*}
This concludes the proof of Lemma~\ref{lem:3-bidders} for irregular distributions.

\section{LAP for Correlated Values}
\label{sec:correlated}

In this section we show that when the value distributions are correlated, then LAP cannot achieve better performance than LA, even in the case of $n=2$ bidders. 
We construct a correlated distribution with the following unnatural \emph{``cryptographic"} feature: the seller can precisely recover the value of the lower bidder from the bid of the higher bidder, while the lower bid reveals almost no useful information about the top value. 
We construct the distribution of pairs of values $(v_1\le v_2)$ as follow.
\begin{itemize}
\item $v_1\sim \disti[1]$, where $\prob{v_1>x}=\frac{1}{x}$ is a discrete equal revenue distribution supported on the interval $[1,\frac{1}{\eps_1}]$ for a small $\eps_1>0$. 
\item $v_2= v_1\cdot (\xi_2+\eps^3)$, where $\eps>0$ is a negligibly small number, $\xi_2\sim\disti[2]$ and $\disti[2]$ is a discrete equal revenue distribution with the support $[1,\frac{1}{\eps_2}]$ for a small $\eps_2>0$. We assume that $\eps_1 \ll \eps_2$.
\item The supports of the discrete distributions $v_1\sim\disti[1],\xi_2\sim\disti[2]$ consists of numbers that are multiples of $\eps<\eps_1$. This allows us to recover $v_1$ from the value of $v_2$ alone. Indeed, notice that $v_2=n_1\eps\cdot(n_2\eps+\eps^3)$ for integer $n_1,n_2<\frac{1}{\eps^2}$. Then $\frac{v_2}{\eps^2}=n_1\cdot n_2 +\eps^2\cdot n_1$ and its fractional part $\left\{\frac{v_2}{\eps^2}\right\}=v_1\cdot\eps$. On the other hand, we can ignore the tiny term of $\eps^3$ in $v_2$, i.e., just think of $v_2=v_1\cdot \xi_2$ without any noticeable change to the posterior distribution $\disti[2](v_2 | v_1)$.
\end{itemize}

\begin{theorem}
\label{thm:correlated}
No LAP mechanism achieves better than $\frac{1}{2}+o(1)$ approximation to the optimal revenue for the above instance of correlated distribution.
\end{theorem}
\begin{proof}
To simplify the presentation we will ignore in the analysis small terms that appear due to the discretezation by $\eps$. Assume to the contrary that there is a LAP mechanism with better than $\frac{1}{2}+o(1)$ approximation to the optimum. This mechanism keeps increasing the lowest threshold for both bidders in leaps $[s_1,t_1], [s_2,t_2],\ldots$ until at least one bidder drops, in which case it offers a new final price to the surviving bidder, if a single one survives, or sells the item at the latest offered price to one of the bidders picked uniformly at random, if none of them survived.
Here we assume that for every $i$, the LAP proceeds in one jump from $s_i$ to $t_i$ and then increases the price continuously from $t_i$ to $s_{i+1}$.

The optimum mechanism can offer to the second bidder the highest possible price of $\frac{v_1}{\eps_2}$ and if he rejects, offer a price of $v_1$ (which the seller can identify from the second bidder) to the first bidder. Thus the optimum gets the expected revenue of $(1-\eps_2)v_1+\frac{v_1}{\eps_2}\eps_2=(2-o(1))v_1$ for each fixed value of $v_1$, which in expectation gives\footnote{We write $\approx$ (and not ``$=$''), as the distributions in our construction are discrete instead of the continuous equal revenue distribution. We write the integral and not the summation over the discrete values to simplify our presentation and calculations. This is w.l.o.g., as $\eps_1,\eps_2 \to 0$ and the integral above arbitrarily closely approximates the respective summation. We apply the same  approximation principle in our derivations below.} 
\[\opt \approx \int_{1}^{1/\eps_1}\frac{1}{t^2}\cdot(2-\eps_2)\cdot t\;dt = (2-\eps_2)\ln(1/\eps_1), \]
Next, we analyse the performance of LAP in different cases of values  $v_1,v_2$. We will argue that in almost all cases LAP's revenue is not more than the revenue of an LA which is a 2-approximation to $\opt$ (LA can only sell to the second bidder and thus for each value of $v_1$ the best revenue it can get is $v_2$). 

First, if $v_1$ falls into any interval $(t_i,s_{i+1})$, then LAP will certainly loose the lowest bidder. The best revenue it can get from $v_2= v_1\cdot \xi_2$ will be $v_1$. I.e., in expectation it is $\int_{t_i}^{s_{i+1}}\frac{1}{t^2}\cdot t\;dt$.

Second, if $v_1\in[s_i,t_i]$ for $t_i<\frac{1}{\eps_1}$. Then there are two possibilities for the top bid $v_2$: (i) $v_2<t_i$, in which case LAP sells equally likely to one of the bidders at price $s_i$; (ii) $v_2\ge t_i$, in which case the LAP can only sell to the highest bidder and its allocation and payment rules are bounded by the incentive constraint from case (i). In case (i) the revenue is $s_i$, while in (ii) case with half probability the item must be sold to the first bidder at a low price of $s_i$ (due to the incentive constraint for the top bidder) and with the remaining probability it can be sold at the best price of $t_i$ (since $\disti[2]$ is the equal revenue distribution). Hence, the total revenue of LAP in this case is 
\begin{equation}
\label{eq:lap_correlated_interval}
\lap([s_i,t_i])=\int_{s_i}^{t_i}\frac{1}{t^2}\cdot
\left(s_i\cdot\frac{t_i-t}{t_i}+\left(\frac{s_i}{2}+\frac{t_i}{2}\right)\cdot\frac{t}{t_i}\right)\;dt
=1-\frac{s_i}{t_i}+\frac{t_i-s_i}{2t_i}\ln \frac{t_i}{s_i}.
\end{equation}
Denoting by $x\eqdef\frac{s_i}{t_i}\in(0,1]$ we get  $\eqref{eq:lap_correlated_interval}=1-x-\frac{1-x}{2}\ln x \le -\ln(x) $ (the latter inequality can be verified using Taylor expansion of $\ln (1- y)$ where $y=1-x$). Thus $\eqref{eq:lap_correlated_interval}\le \int_{s_i}^{t_i}\frac{1}{t^2}\cdot t\;dt=\la([s_i,t_i]).$

Finally, if $v_1\in [s_i,t_i]$ for $t_i>\frac{1}{\eps_1}$, then LAP cannot get more revenue than $v_1$ from the second bidder and the payment from the lower bidder is not more than $s_i$. The total expected revenue of LAP in this case is 
\[
\lap([s_i,1/\eps_2])\le\int_{s_i}^{1/\eps_1}\frac{1}{t^2}\cdot\left(s_i+t\right)\;dt
\le 1 + \la([s_i,1/\eps_1]).
\]

In conclusion we get that $\lap\le 1 + \la\le 1+\opt(\frac{1}{2}+o(1))\le\opt(\frac{1}{2}+o(1))$, since $\opt=(2-\eps_1)\ln(1/\eps_2)$.
\end{proof}

This cryptographic construction makes our correlated distribution not quite realistic. The point we illustrate here is that the class of general correlated distribution even for $n=2$ bidders is just unrealistically general.

\bibliographystyle{splncs04}
\bibliography{LAP}

\newpage

\appendix

\section{Truthfulness of \lap}
\label{sec:truthful-lap}
\begin{proofof}{Theorem~\ref{th:LAP_truthful}}
If $i$ is a bidder with $\alloci(\vals)>0$, then she always has an option either not to buy the item, or to pay $\lowc\le\vali$ whenever she gets the item in a lottery, i.e., $\pricei\le \alloci\cdot \vali$. Thus $\utili=\vali\cdot\alloci-\pricei\ge 0$ and $\lap$ is IR.  

Next we show that any LAP is DSIC. We assume that all bidders bid truthfully and argue that no bidder $i$ can improve their utility by bidding $\bidi\ne \vali$. 

We consider the bidders who dropped strictly before the last increment (either continuous or discrete from $\lowc$ to $\upc$) of the cut-off price $c$. Let $c'\ge \vali$ be the value of the cut-off price, when bidder $i$ has dropped. If $\bidi<c'$, then nothing changes for the bidder $i$, as $\alloci=0$. If $\bidi>c'$ and $i$ gets a different allocation $\alloci>0$, then $i$ has to pay at least $\pricei\ge c'\cdot\alloci$  which gives a non-positive utility to bidder $i$.

Next, consider the case when $\lap$ stops at the cut-off price $c'$ while continuously increasing $c$. Then only one bidder $i^*$ whose value $\vali[i^*]\ge c'$ may receive the item. Let $\upc<c'$ be the value of the cut-off price right after the last discrete jump of $c$ ($\upc=0$ if there were no discrete jumps). We have already discussed above that any bidder $i$ with 
$\vali<c'$ cannot get positive utility. The bidder $i\ne i^*$ with $\vali=c'$ drops the last from $S$. Les us consider her possible bids: (i) she would loose and get $\alloci=0$ by bidding $\bidi<\upc$ or $\bidi\in[\upc,c']$; (ii) she would have to pay $\pricei> c'\cdot \alloci =\vali\cdot\alloci$ by bidding any $\bidi>c'$. Consider all possible bids of $i^*$: (i) $i^*$ looses to $i$ by bidding $\bidi[i^*]<\upc$ or by dropping earlier at any earlier cut-off price $c$; (ii) $i^*$ gets the same take-it-or-leave-price  $p(\valsmi[i^*])$ with any bid $\bidi[i^*]\ge c'$. In summary, any bidder $i\in N$ cannot increase their utility by bidding $\bidi\ne\vali$ if $\lap$ stops while continuously increasing $c$.

We are only left to consider the case when $\lap$ stops (and potentially allocates the item) within a discrete increment of $c$ from $\lowc$ to $\upc>\lowc$.
We also only need to consider the bidders in $S(\lowc)$, since any other bidder drops before the last increment of $c$. 

Observe that any bidder $i\in S(\lowc)$ ($\vali\ge\lowc$) would loose a chance to win the item if they bid $\bidi < \lowc$, so such deviation results in $0$ utility for them. 

Next, let $i$ be any bidder with $\lowc\le\vali<\upc$. Any bid $\lowc\le\bidi<\upc$ leads to the same outcome and payment for $i$ as $\bidi=\vali$.

If $i$ bids $\bidi\ge \upc$, then either two bidders will get to the next stage $c>\upc$ (in this case $i$ will need to pay at least $\upc>\vali$), or only $i$ accepts the cut-off price $c=\upc$ (then $i$ would take the first option out of two, which will give her the same allocation of $\alloci=\frac{1}{|S(\lowc)|}$ and payment $\pricei=\frac{\lowc}{|S(\lowc)|}$ as if $\bidi=\vali$).
Finally, there could be at most one bidder $i^*$ with value $\vali[i^*]\ge\upc$. Note that bidding any $\bidi[i^*]\ge\upc$ results in the same choice of two options for $i^*$, while bidding below $\lowc$ results in $\alloci[i^*]=0$. Any bid $\lowc\le\bidi[i^*]<\upc$ leads to an allocation probability $\frac{1}{|S(\lowc)|}$, i.e., $i^*$ would be restricted to choosing the first option. 

We conclude that truthful bidding is a dominant strategy of every bidder.
\end{proofof}
\newpage
This project has received funding from the European Research Council (ERC) under the European Union's Horizon 2020 research and innovation program (grant agreement No. 866132), and by the Israel Science Foundation (grant number 317/17)

\end{document}